\documentclass[10pt,a4paper]{article}
\usepackage{amsmath}
\usepackage{amsfonts}
\usepackage{amssymb}
\usepackage{graphicx}
\usepackage[ruled]{algorithm2e}
\usepackage{booktabs}
\usepackage{tabularx}
\usepackage{lscape}

\newtheorem{te}{Theorem}[section]
\newtheorem{lemma}{Lemma}[section]
\newtheorem{remark}{Remark}[section]
\newtheorem{pro}{Proposition}[section]

\newenvironment{proof} {\par \noindent \textbf{Proof: }}{\QED\par \bigskip \par}
\newcommand{\QED}{\hfill$\square$}
\newcommand{\rz}{\vspace{0.1cm}}

\title { \bigskip
    \bf {Constructions of hamiltonian graphs with
    bounded degree and diameter $O (\log n)$}
    \thanks{Supported by the Research Program P1-0285 of Slovenian Agency for Research and the Grant 144007 of Serbian Ministry of Science and Technological Development.}
}

\author
{
{\large \sc Aleksandar Ili\' c\footnotemark[3]} \\
{\em \normalsize Faculty of Sciences and Mathematics, University of Ni\v s, Serbia} \\
{\normalsize e-mail: { \tt aleksandari@gmail.com }} \and
{\large \sc Dragan Stevanovi\' c} \\
{\em \normalsize University of Primorska---FAMNIT, Glagolja\v ska 8, 6000 Koper, Slovenia,} \\
{\em \normalsize Mathematical Institute, Serbian Academy of Science
and Arts,}
{\em \normalsize Knez Mihajlova 36, 11000 Belgrade, Serbia }\\
{\normalsize e-mail: { \tt dragance106@yahoo.com}} }

\begin{document}

\maketitle

\vspace{-0.5cm}

\begin{abstract}
    Token ring topology has been frequently used in the design of
    distributed loop computer networks and one measure
    of its performance is the diameter. We
    propose an algorithm for constructing hamiltonian graphs with $n$
    vertices and maximum degree $\Delta$ and diameter $O (\log n)$,
    where $n$ is an arbitrary number. The number of edges is asymptotically bounded by
    $( 2 - \frac{1}{\Delta - 1} - \frac{(\Delta - 2)^2}{(\Delta - 1)^3}) n$.
    In particular, we construct a family of hamiltonian
    graphs with diameter at most $2 \lfloor \log_2 n \rfloor$,
    maximum degree $3$ and at most $1+11n/8$ edges.
\end{abstract}


        \footnotetext[3]
        {
        Corresponding author. If possible, send your correspondence via e-mail.
        Otherwise, postal address is:
        Department of Mathematics and Informatics, Faculty of Sciences and Mathematics,
        Vi\v segradska 33, 18000 Ni\v s, Serbia
        }

\section{Introduction}


An undirected graph $G = (V, E)$ can be used as a mathematical model
for computer networks, where $V$ is the set of vertices and $E$ is
the set of edges. The number of edges adjacent to a vertex $v$ is
called the degree of the vertex~$v$. A graph is regular if all
vertices have equal degrees. The distance $d (v, u)$ between two
vertices $v$ and $u$ is the number of edges on a shortest path
between $v$ and $u$. The diameter $D$ of the graph is the maximum
distance between any pair of vertices: $D = \max_{v, u \in V} d(v,
u)$. A cycle is a sequence of three or more vertices such that two
consecutive vertices are adjacent and with no repeated vertices
other than the start and end vertex. A \emph{hamiltonian cycle} is a
cycle that visits each vertex of a graph exactly once. A graph $G$
is $1$-hamiltonian if, after removing an arbitrary vertex or an
edge, it still remains hamiltonian. A $1$-hamiltonian graph G is
optimal if it contains the least number of edges among all
$1$-hamiltonian graphs with the same number of vertices as $G$. \rz


Networks with at least one ring structure (hamiltonian cycle) are
called loop networks. Distributed loop networks are extensions of
ring networks and are widely used in the design and implementation
of local area networks and parallel processing architectures. There
are many mutually conflicting requirements when designing the
topology of a computer network. For example, no pair of processors
should be too far apart in order to support efficient parallel
computation demands. The hamiltonian property is one of the major
requirements. The \emph{token passing} is a channel access method
where data is transmitted sequentially from one ring station to the
next with a control token circulating around the ring controlling
access. \rz

An open problem considered in a survey $\cite{bermond95}$ on
distributed loop networks is following: Find hamiltonian networks,
$\Delta$-regular on $n$ vertices with a diameter of order $O (\log
n)$. This problem is related to the famous $(n, \Delta, D)$ problem
in which we want to construct a graph of $n$ vertices with maximum
degree $\Delta$ such that the diameter $D$ is minimized, but
hamiltonicity is not an issue. The lower bound on the diameter $D$
is called the \emph{Moore bound},
$$D \geqslant \log_{\Delta - 1} n - \frac{2}{\Delta}.$$


Harary and Hayes $\cite{harary96}$ presented a family of optimal
$1$-hamiltonian planar graphs on $n$ vertices. Wang, Hung and Hsu
$\cite{wang98a}$ presented another family of optimal $1$-hamiltonian
graphs, each of which is planar, hamiltonian, cubic, and of diameter
$O (\sqrt{n})$. In the literature three other families of cubic,
planar and optimal $1$-hamiltonian graphs with diameter $O (\log n)$
are described. These constructions are possible only for special
choices of $n$, as shown in Table \ref{optimalGraphs}. \rz

\begin{table}[ht]
\centering \small
    \begin{tabular}{ l l l l l }
    \toprule

    Reference & Name & $n$ & $D$ & Comment \\[0.5ex]
    \midrule
    Wang et al. $\cite{wang98}$ & Eye graph & $6 \cdot 2^s - 6$ & $O (\log n)$ & $3$-connected \\
    Hung et al. $\cite{hung99}$ & Christmas tree & $3 \cdot 2^s - 2$ & $2s$ & hamiltonian connected \\
    Kao et al. $\cite{kao03}$ & Brother tree & $6 \cdot 2^s - 4$ & $2s + 1$ & bipartite
    \\ [0.5ex]
        \bottomrule
    \end{tabular}
\label{optimalGraphs}

\caption{Families of $1$-hamiltonian graphs and diameter $O (\log
n)$}
\end{table}

The best constructions for cubic graphs have diameter $1.413 \log_2
n$ (see $\cite{capablo03}$). It is shown in $\cite{bollobas88}$ that
a cubic graph obtained by adding a random perfect matching to a
cycle has a diameter of order $O (\log n)$. In the same paper, the
authors proved the following result:

\begin{te}
    Suppose $T$ is a complete binary tree on $2^k - 1$ vertices. If
    we add two random matchings of size $2^{k - 2}$ to the leaves of
    $T$, then the resulting graph has diameter $D$ satisfying
    $$\log_2 n - 10 \leqslant D (G) \leqslant \log_2 n + \log_2
    \log_2 n + 10$$
    with probability approaching $1$ as $n$ approaches infinity.
\end{te}

In $\cite{robinson92}$ it is shown that almost all $k$-regular
graphs are hamiltonian for any $k \geqslant 3$, by an analysis of
the distribution of $1$-factors in random regular graphs. \rz


In this paper we propose an algorithm for constructing hamiltonian
graphs with $n$ vertices, maximum degree $\Delta$ and diameter
bounded by:
$$D \leqslant 2 \cdot \lfloor \log_{\Delta - 1} n \rfloor.$$

Our main contribution is that we assure diameter $O (\log_{\Delta -
1} n)$ for every~$n$, not just for special values, while
hamiltonicity and small diameter are achieved by using significantly
less edges. \rz

The paper is organized as follows. In Section 2 we describe the
linear algorithm and prove that it produces hamiltonian graphs of
degree at most~$\Delta$ and diameter at most $2\log_{\Delta-1} n$.
In Section 3 we deal with the case $\Delta = 3$ and improve the
upper bound for the number of edges in the resulting graph to
$\lfloor \frac{11n+6}{8}\rfloor$. The same approach may be applied
for $\Delta > 3$ and this way we get a hamiltonian graph with the
average degree asymptotically equal to $4 - \frac{4}{\Delta - 1}$.
In the end, we show that the algorithm may be modified in such a way
that it constructs a planar hamiltonian graph with degree at most
$\Delta$ and diameter at most $2 \lfloor \log_2 n \rfloor$ and also
point out some experimental results for diameter when $\Delta = 3$.

\section{The algorithm}

A complete binary tree is a tree with $n$ levels, where for each
level $d \leqslant n - 1$, the number of the existing nodes at level
$d$ equals $2^d$. This means that all possible nodes exist at these
levels. An additional requirement for a complete binary tree is that
for the $n$-th level, while not all nodes have to exist, the nodes
that do exist must fill the level from left to right (for more
details see \cite{cormen01}). A complete binary tree is one of the
most important architectures for interconnection networks
\cite{leighton92}. \rz

A generalization of binary trees are $\Delta$-ary trees. Namely,
every vertex has at most $\Delta - 1$ children, and all vertices
that are not in the last level have exactly $\Delta - 1$ children.
All vertices in the last level must occupy the leftmost spots
consecutively. \rz

We propose the following algorithm for constructing hamiltonian
graphs of order $n$, maximum degree $\Delta$ and diameter $\leq 2
\log_{\Delta - 1} n$. First we construct a complete $\Delta$-ary
tree with $n$ vertices, and for every vertex we store its degree,
parent, and all children from left to right. A labeled $\Delta$-ary
tree contains labels $1$ through $n$ with root being $1$, branches
leading to nodes labeled $2, 3, \ldots, \Delta$, branches from these
leading to $\Delta + 1, \Delta + 2, \ldots, (\Delta - 1)(\Delta - 1)
+ \Delta$, and so on. We also maintain a queue of all leaves. \rz

\begin{algorithm}
    \small
    \KwIn{$\Delta$-ary rooted tree with $n$ nodes}
    \KwOut{Hamiltonian cycle in an array $cycle$, $m$ is the total number of edges}

    $v = root$\;
    $k = 1$\;
    $m = n - 1$\;
    \While{$k \leqslant n$} {
        $cycle [k] = v$\;
        $mark [v] = true$\;
        \uIf{$parent [v]$ is not marked} {
            $v = parent [v]$;
        }
        \uElseIf{ there exists a child $u$ from the set of neighbors which is not marked} {
            $v = u$;
        }
        \Else {
            Let $u$ be a random unmarked leaf\;
            Add edge $(v, u)$\;
            $m = m + 1$\;
            $v = u$\;
        }
        $k = k + 1$\;
    }
    Add edge $(v, root)$\;
    $m = m + 1$\;
    \caption{Constructing a hamiltonian supergraph of a complete $\Delta$-ary tree}
\end{algorithm}

We will traverse vertices in order to form a hamiltonian cycle. The
starting vertex is the root of the tree. First we try to go up the
tree through the parent of the current vertex - if the parent is
already visited, we choose one of its children. In the case where
all neighbors of the current vertex are marked, we pick an arbitrary
unmarked leaf and add an edge connecting these two vertices.

\begin{figure}[h]
  \center
  \includegraphics [width = 8cm]{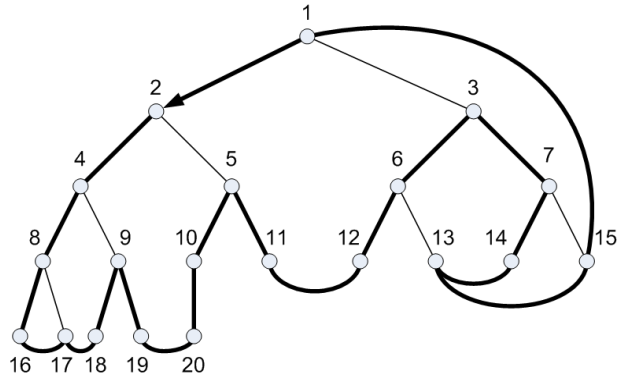}
  \caption { \textit{Example of algorithm execution for a graph with $n = 20$, $\Delta = 3$ and vertex $1$ as the root} }
  \label{example}
\end{figure}

\begin{te}
\label{th-2.1}
    A graph constructed by the above algorithm is hamiltonian.
\end{te}

\begin{proof}
The basic idea of the algorithm is to traverse the hamiltonian path
by adding edges when they are needed. In the end, we join the last
visited vertex to the root of the binary tree. There can be at most
one vertex other than the root with a degree greater than $1$ and
less than $\Delta$ in the tree - let this vertex have the label $s$.
All edges that we add during the execution of the algorithm connect
either two leaves or a leaf and $s$. By symmetry, we can eliminate
the vertex $s$ by traversing from the root to $s$ in the first few
steps of the algorithm. Therefore, this algorithm does not increase
the maximum degree $\Delta$ in the graph. \rz

We cannot visit any vertex twice, so we have to prove that all
vertices are marked. Assume that the vertex $v$ is not visited and
the algorithm has finished. If any vertex in its subtree is visited,
then we would have already visited~$v$, because we first try to go
upwards. Therefore, all vertices in its subtree are unvisited. There
is at least one leaf in the subtree, and we have to visit this leaf
during the execution of the algorithm. This is a contradiction, so
the constructed graph is hamiltonian.
\end{proof}

The diameter of this graph is less than or equal to twice the depth
of the tree.
$$D (G) \leqslant 2 \cdot \lfloor \log_{\Delta - 1} n \rfloor.$$

\begin{lemma}
    The number of leaves $L (n)$ in a complete $\Delta$-ary tree is
    $$L (n) = \left \{
    \begin{array}{ll}
        n - \lfloor \frac {n}{\Delta - 1} \rfloor , & \quad
        \mbox{ if } n \equiv 1 \pmod{ (\Delta - 1) } \\
        n - 1 - \lfloor \frac {n}{\Delta - 1} \rfloor , & \quad
        \mbox{ if } n \not \equiv 1 \pmod{ (\Delta - 1) } \\
    \end{array} \right.$$
\end{lemma}

\begin{proof}
    We know that $L (1) = 1$ and $L (k) = k - 1$ for $k = 2, 3,
    \ldots, \Delta - 1$. Whenever we properly add $\Delta - 1$ vertices
    to this tree, we always get exactly $\Delta - 2$ new leaves. This
    proves the recurrent relation for the number of leaves
    $$L (n) = L (n - \Delta + 1) + \Delta - 2.$$

    By mathematical induction, one can easily prove the explicit formula
    for~$L (n)$.
\end{proof}

\begin{te}
    The number of edges in the constructed graph is less than
    $(2 - \frac{1}{\Delta - 1}) n + \frac {\Delta - 3}{2}$.
\end{te}

\begin{proof}
    There are $\frac{n\Delta}{2}$ edges in a $\Delta$-regular
    graph. The internal nodes in a $\Delta$-ary tree are nodes with degree greater than one.
    After running the algorithm, every leaf will have degree at most
    three and every internal node will have degree at most $\Delta$.
    This gives an upper bound for the number of edges in the hamiltonian graph
    constructed by the algorithm:
$$
    |E| < \frac{n \Delta}{2} - \frac {\Delta - 3}{2} \left
    ( n - \left \lfloor \frac{n}{ \Delta - 1} \right \rfloor - 1 \right
    ) \leqslant \frac{3n}{2} + \frac {\Delta - 3}{2} \left (\frac
    {n}{\Delta - 1} + 1 \right ) = \left(2 - \frac{1}{\Delta -
    1} \right) n + \frac{\Delta - 3}{2}.
$$
\end{proof}

This bound is less than $\frac{5n}{3}$ for $\Delta = 4$, and less
than $2n$ if $\Delta > 4$ and the depth of the tree is greater
than~two. Based on this fact, we give the following

\begin{pro}
Time and memory complexity of the proposed algorithm is linear $O
(n)$. \QED
\end{pro}

\section{Number of edges added for $\Delta=3$}

Cubic graphs are of special interest in token ring topologies and
because of their importance, in this section we improve the previous
estimations for the number of edges added in the construction of a
hamiltonian path for the case $\Delta = 3$. \rz

First we examine graphs with $n = 2^{k + 1} - 1$ vertices, where $k
\geqslant 1$. A corresponding binary tree has complete last level
and there are $2^k$ leaves in the tree. Let $f (k)$ be the number of
additional edges added to this tree, when we start the execution of
the algorithm in an arbitrary leaf and end it in another leaf. One
can easily verify that $f (1) = 0$ and $f (2) = 2$. In addition, we
define $f(0) = 0$. After traversing upwards to the root and then
downwards to an arbitrary leaf, the graph induced by the unvisited
vertices are disjoint complete binary trees of heights $0, 1,
\ldots, k - 2$. Therefore, we may write the recurrent formula:
$$f (k) = 2 (k - 1) + 2 \cdot \sum_{i = 0}^{k - 2} f (i).$$

By strong induction we will prove that:
$$f (k) = \left \{
\begin{array}{ll}
         \frac {2}{3} \left ( 2^k - 1 \right ), & \quad \mbox{if $k$ is even} \\
         \frac {2}{3} \left ( 2^k - 2 \right ), & \quad \mbox{if $k$ is odd} \\
\end{array} \right.$$

Assume that the formula holds for all numbers $m < 2k$, and now we
will prove it for $2k$ and $2k + 1$.
\begin{flalign*}
f (2k) &= 2 (2k - 1) + 2 \cdot \sum_{i = 0}^{k - 1} f (2i) + 2 \cdot
\sum_{i = 1}^{k - 1} f (2i - 1)
= 4k - 2 + \frac{4}{3} \cdot \sum_{i = 0}^{2k - 2} 2^i -
\frac{4}{3} \cdot (3k - 2) = \frac{2}{3} \cdot (2^{2k} - 1). \\
f (2k + 1) &= 2 (2k + 1 - 1) + 2 \cdot \sum_{i = 0}^{k - 1} f (2i) +
2 \cdot \sum_{i = 1}^{k} f (2i - 1)
= 4k + \frac{4}{3} \cdot \sum_{i = 0}^{2k - 1} 2^i - \frac{4}{3}
\cdot 3k = \frac{2}{3} \cdot (2^{2k + 1} - 2).
\end{flalign*}

To estimate the actual number of additional edges in the proposed
algorithm we have to start from the root and add an additional edge
connecting the last visited vertex and the root of the binary tree.
After visiting the first leaf we have that the unvisited vertices
form disjoint complete binary trees of heights $0, 1, 2, \ldots, k -
2$ and $k - 1$. So, the total number of edges to be added equals:
$$f (0) + f (1) + \ldots + f (k - 1) + k = \frac {f (k + 1)}{2} =
\left \lfloor \frac{n}{3} \right \rfloor.$$

Therefore, we proved the following result.
\begin{te}
    The number of additional edges in the case of a complete binary tree
    with $2^k - 1$ vertices does not depend on the choice of random
    leaves in the algorithm and this number equals $\frac {f (k +
    1)}{2}$. \QED
\end{te}

Now assume that $n$ is not of the form $2^{k + 1} - 1$. The main
result in this section is the following theorem.

\begin{te}
    For every integer $n \in N$, there exists a hamiltonian
    graph with diameter at most $2 \lfloor \log_2 n \rfloor$,
    maximum degree $3$ and at most $\lfloor \frac{11n}{8} + \frac{3}{4} \rfloor$ edges.
\end{te}

\begin{proof}
The number of leaves in a corresponding binary tree with $n$
vertices is $\lceil \frac {n}{2} \rceil$. Some leaves are on the
last level, and some are on the level before last. Consider the
consecutive leaves in the last level starting from the left and
group them into groups of size four. This way we get subtrees of
height two and there can be at most three unpaired leaves. We will
do the same thing for the leaves in the previous level, but starting
from the right. \rz

Therefore, the number of subtrees of height two is less than or
equal to $\frac{n}{8}$ and greater than $\frac{n - 2 \cdot 3}{8}$.
According to the two cases in Figure \ref{tree4}, there can be two
or three leaves in each of these subtrees that will get and keep
degree two until the algorithm finishes. Thus, the number of edges
we do not have to add equals half the number of vertices of degree
two, and it is between $\frac {n}{8} - \frac{3}{4}$ and
$\frac{3n}{16}$. The total number of edges in the constructed graph
after running the algorithm will be between $\lfloor \frac{21 n}{16}
\rfloor$ and $\lfloor \frac{11 n}{8} + \frac{3}{4} \rfloor$.
\end{proof}

\vspace{-3mm}

\begin{figure}[h]
  \center
  \includegraphics [width = 8cm]{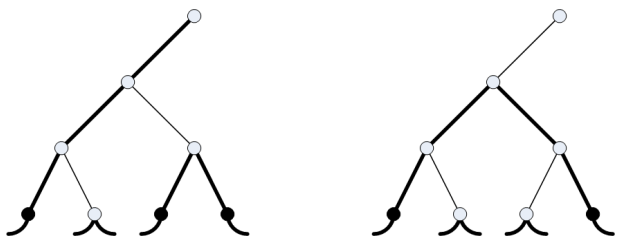}
  \caption { \textit{Two cases for subtrees of height two} }
  \label{tree4}
\end{figure}

\begin{remark}
After the hamiltonian cycle is constructed, we may insert additional
edges into the graph to make it cubic, provided $n$ is even.
\end{remark}

For the case $\Delta > 3$, one can construct a hamiltonian graph
with at most
$$
\left( 2 - \frac{1}{\Delta - 1} - \frac{(\Delta - 2)^2}{(\Delta -
1)^3} \right) n + \frac{\Delta - 3}{2} + 2 (\Delta - 2)
$$
edges. The same approach may be applied to complete binary subtrees
of greater heights to obtain a slightly finer bound for the total
number of edges.

\section{Concluding remarks}

The proposed algorithm may be modified to construct a planar
hamiltonian graph.

\begin{te}
    By appropriately choosing descendant and unvisited
    leaves in the algorithm, one can assure that the constructed graph is planar.
\end{te}

\begin{proof}
    In order to prove the theorem, we will construct a hamiltonian path that
    starts at the leftmost leaf and ends in the nearest leaf (neighboring
    leaf or leaf that is at distance three from it). For small values
    of $n$, this can be easily verified. In the general case
    we first go upwards to the root and then to the rightmost leaf. Now,
    the $\Delta$-ary tree is partitioned into smaller trees, which will be
    traversed by induction from left to right. These binary trees
    do not have vertices in common, and we can independently add
    necessary edges which do not intersect the existing edges.
    We reduce our problem to the previous case by going to the leftmost
    leaf. Now we have disjoint trees and we traverse them starting
    from the leaves. Finally, we add the last edge without intersection problems
    (as in Figure~\ref{example}).
\end{proof}


For the case $\Delta = 3$, in our implementation we always choose a
leaf that is farthest from the current leaf. This heuristic is done
by the breadth first search. We use only three arrays of length $n$,
so memory requirements are linear in $n$. Time complexity is $O
(n^2)$, because the number of edges is $m \leqslant \frac{3n}{2} = O
(n)$. The diameters of the examples of cubic graphs constructed by
this algorithm are shown in Figure \ref{graphics}, where the
$x$-axis carries $n/2$ and $n$ is the number of nodes. It is obvious
from this figure that the constant~$2$ from our bound is not the
best possible.
\newpage

\begin{figure}[h]
  \center
  \includegraphics [width = 10.5cm]{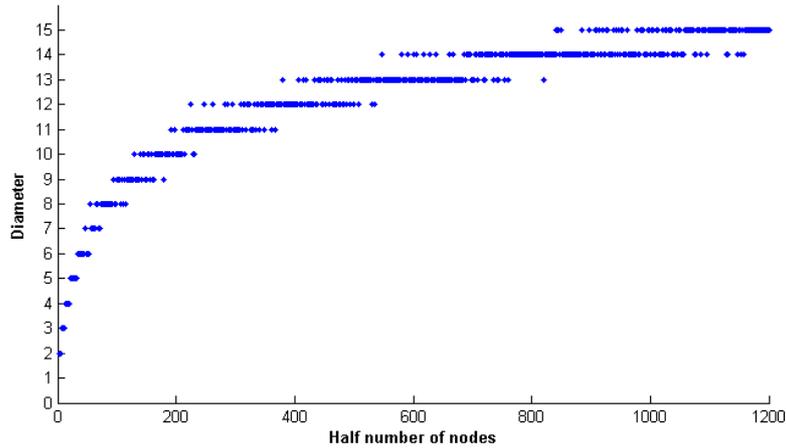}
  \caption { \textit{The size of diameter for $n = 4$ to $n = 2400$ for $\Delta = 3$} }
  \label{graphics}
\end{figure}

Instead of choosing the leaves at random, we can do it more
sophisticatedly and further reduce the diameter of the graph. One
possible way is to add a matching which connects only the leaves on
the last level in the root's left and right subtree. This way we can
still apply the algorithm, but if we have to choose a random leaf to
continue---we first check whether the paired leaf in the matching is
marked. This way we decrease the diameter by a constant, which is at
least one. We leave for future study to see whether this approach
may be used to construct $1$-edge hamiltonian graphs. \vspace{0.2cm}

{\bf Acknowledgement: } The authors are grateful to the reviewers
for their valuable comments and suggestions.

\end{document}